\documentclass{article}
\usepackage{spconf,color,dsfont,bm,epsfig,graphicx,cite,algorithmicx,algorithm}
\usepackage{amsmath,amsfonts,amsthm,amssymb} % Math packages
\usepackage{stfloats}
\usepackage{epstopdf}

\newcounter{mytempeqncnt}

\newtheorem{pro}{Proposition}

\long\def\symbolfootnote[#1]#2{\begingroup
\def\thefootnote{\fnsymbol{footnote}}
\footnote[#1]{#2}\endgroup}

\begin{document}
\title{Optimal Time Reuse in Cooperative D2D Relaying Networks}
\name{Zhaowei Zhu$^1$, Shengda Jin$^1$, Yu Zeng$^2$, Honglin Hu$^2$, and Xiliang Luo$^1$}
\address{$^1$ShanghaiTech University, Shanghai, China,
$^2$Shanghai Advanced Research Institute, CAS, Shanghai, China\\
Email:{ {luoxl}@shanghaitech.edu.cn}}

\maketitle

\ninept

\begin{abstract}
Device-to-device (D2D) communication has become one important part of the
5G cellular networks particularly due to the booming of proximity-based applications,
e.g. D2D relays. However, the D2D relays may create strong interference to nearby users.
Thus interference management in a cellular network with D2D relays is critical.
In this paper, we study the optimal time reuse patterns in a cellular network with
cooperative D2D relays and derive the corresponding optimized relaying strategies.
Due to the binary association constraints and the fact that the total number of feasible
time reuse patterns increase exponentially with the number of servers,
we have a large scale integer programming problem which is formidable to solve.
However, we show that we just need to activate a few time reuse patterns in the order
of the number of users. Accordingly, a low-complexity algorithm is proposed to find
the set of active time reuse patterns and solve this problem in an approximated
way. Numerical simulations demonstrate that our scheme can efficiently allocate the time
resources and determine the relaying strategies. Furthermore, the proposed scheme
offers significant gains relative to the existing ones.
\end{abstract}


\section{Introduction}

With the ever-increasing mobile data traffic demands, the proximity-based communication
through the device-to-device (D2D) links is regarded as one promising technology in 5G
cellular networks \cite{2014_survey}. Different from the base station (BS)-based links,
D2D links usually appear in short-range scenarios and can significantly enhance
the network throughput performance without interfering distant users
\cite{2016_Xu_D2D}. In a cooperative D2D network, when both types of links share the same
time or spectrum resources, severe interference can still occur among nearby users
without careful resource management \cite{2017_Liu_interference}.

To make full use of the limited resources, a lot of works have been done to optimize
D2D-enabled cellular networks
\cite{2016_Ma_HenNet_D2D,2013_Xu_Cooperation,2013_Zhou_D2D_relay,
2016_Mastronarde_D2D_relay,2015_Li_MMM}.
In particular, the spectrum allocation and power control were jointly optimized in
\cite{2016_Ma_HenNet_D2D} to maximize different network utility metrics under QoS
constraints, while D2D relays were not modeled in \cite{2016_Ma_HenNet_D2D}. Recent works
in \cite{2013_Xu_Cooperation,2013_Zhou_D2D_relay,2016_Mastronarde_D2D_relay,2015_Li_MMM}
emphasized the importance of D2D relays and recommended some D2D user equipments (DUEs)
to serve as
D2D relays for other DUEs experiencing poor BS-based links. Orthogonal resource allocations
and corresponding relaying strategies were optimized jointly under different constraints in
\cite{2013_Xu_Cooperation,2013_Zhou_D2D_relay,2016_Mastronarde_D2D_relay,2015_Li_MMM}.
Specifically, the orthogonal time/frequency resources were allocated between the BS and the
D2D relays in \cite{2013_Zhou_D2D_relay}, among the neighboring nodes in \cite{2015_Li_MMM},
and among the D2D relays in \cite{2013_Xu_Cooperation,2016_Mastronarde_D2D_relay}.
However, we know non-orthogonal resource reuse among the BSs and the D2D relays will
further enhance the network throughput performance.
Each reuse strategy was modeled as one reuse pattern and the frequency resource allocations
were optimized over all possible reuse patterns for heterogeneous networks (HetNets) in
\cite{2015_Zhuang,2016_Zhuang,2016_Kuang}. In particular, the network-wide average packet
delay was minimized in \cite{2015_Zhuang}, the energy efficiency was maximized in
\cite{2016_Zhuang}, and the network proportional fairness (PF) metric was maximized in
\cite{2016_Kuang}.

In this paper, we endeavor to find the optimal resource reuse scheme in a cooperative
network with D2D relays. The major contributions can be summarized as follows.
 Firstly, we study the optimal set of time reuse patterns
and the corresponding relaying strategies such that the PF metric of the network is
maximized. Even though the total number of feasible reuse patterns increases exponentially
with the number of users, we show that we just need to activate a few time reuse
patterns in the order of the number of users to achieve the optimal network performance.
Secondly, we put forward an efficient algorithm of low complexity to identify the set of
reuse patterns to be activated and solve the problem approximately. Numerical results
demonstrate our scheme can offer significant gains relative to the existing solutions.

The remainder of this paper is organized as follows. Section \ref{System_model} describes the
system model. Section \ref{Formulation} formulates the problem and characterizes
the optimal time reuse profile. Section \ref{Algorithms} proposes one low-complexity
algorithm. Numerical results are given in Section \ref{Simulation} and Section
\ref{Conclusion} concludes the paper.

\vspace{3pt}
\noindent{\it Notations}:
Notations $\| \bm{x} \|_0$, $\bm A^T$, ${\sf conv}(\mathcal{A})$, $|\mathcal{A}|$,
and $\bm a \oplus \bm b$ stand for $l_0$-norm of vector $\bm{x}$, the transpose of
matrix $\bm A$, the convex hull of the set $\mathcal{A}$, the cardinality of the set
$\mathcal{A}$, and the element-wise XOR operation of two vectors.
$\mathds{1}\{\cdot\}$ stands for the indicator function
and takes the value of $1$ ($0$) when the specified condition is met (otherwise).
$\bm e_n$ denotes the unit vector with only the $n$-th element being $1$ and all other
elements being $0$.

\section{System Model} \label{System_model}

\begin{figure}[t]
\centering
\includegraphics[width = 0.48\textwidth]{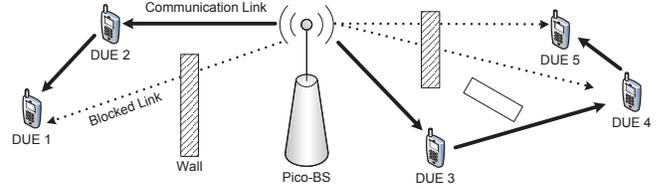}
\vspace{+2pt}
\caption{A cooperative network with D2D relays.
DUE $1$ is served by DUE $2$ since its BS-based link is blocked. Both DUE $4$ and DUE $5$ are served by DUE $3$ for similar reasons.}
\label{Fig:D2D_topology}
\vspace{-10pt}
\end{figure}

As shown in Fig.~\ref{Fig:D2D_topology}, a downlink shared cooperative network \cite{2015_Li_MMM}
is considered here, where a number of nodes coexist including the DUEs and the BSs.
The sets of BSs and DUEs are denoted by $\mathcal{B}$ and $\mathcal{U}$, respectively.
Meanwhile, we define $B:=|\cal{B}|$ and $U:=|\cal{U}|$.
Each DUE is assumed to be able to act as either a relay or a user.
For instance, it can act as a
user to receive data either from the serving BS or from another DUE (behave as one D2D
relay) in one time slot. In another slot, it can act as a relay to help other DUEs.
Thus we assume each DUE can choose to be either a user or a relay in different time slots.
The set of servers $\mathcal{N}$ including both all BSs and D2D relays is defined as
follows.
\begin{equation}
\setlength{\abovedisplayskip}{3pt}
\setlength{\belowdisplayskip}{1pt}
  \mathcal{N} := \{ \underbrace{1,\cdots,B}_{\sf BSs}, \underbrace{{B + 1}, \cdots , {B + U}}_{\sf potential~DUEs} \}.
  \label{Eq:transmitter}
\end{equation}
Let $N :=|\mathcal{N}|$ denote the total number of servers in the set $\mathcal{N}$.
We use a vector $\bm{v}_i = [v_{i,1},\cdots,v_{i,N}]^T$ to indicate the active status
of the servers in one time slot and call it the $i$-th time reuse pattern.
The $n$-th element of $\bm{v}_i$ is $1$, i.e. $v_{i,n} = 1$, when server $n$ is active
under time reuse pattern $\bm v_i$. Otherwise, we have $v_{i,n} = 0$. Note when a server
is active, we assume it utilizes the whole downlink bandwidth for information transmission.
There are totally $2^N-1$ feasible time reuse patterns to be considered and we do not
consider the all $0$ pattern, which means nobody transmits.

Let $\mathcal I = \{1,\cdots,2^N-1\}$ denote the set of indices of all the feasible time
reuse patterns. Furthermore, the set of all time reuse pattern is denoted as
$\tilde{\mathcal{V}} = \{ \bm{v}_i | i \in \mathcal I \}$.
We also use $\mathcal{A}_i$ to denote the set of active servers under pattern $i$,
i.e. $\mathcal{A}_i:= \{ n | v_{i,n} = 1, n \in \mathcal N \}$.
Let $x_i$ denote the fraction of the total time duration of length $T$ allocated to
pattern $i$. The overall time reuse profile is thus given by the vector:
$\bm{x} := [x_1,\cdots,x_{2^N-1}]^T$.
In order to get the optimal network performance, we need to utilize the whole time
duration, i.e. $\sum_{i\in \mathcal I} x_i = 1$.
Furthermore, we assume time-division multiple-access (TDMA) and each server only transmits
data to one served user at one time. A server needs to divide the allocated time resource
under pattern $i$ orthogonally among its served users.
We use $y_{u,n,i} \geq 0$ to denote the fraction of the total duration $T$ allocated to
user $u$ by server $n$ under reuse pattern $i$. Clearly, it is bounded as
$\sum_{u\in \mathcal U} y_{u,n,i}\leq x_i$.
Additionally, as in the current cellular networks, each user is allowed to be associated
to a single server under each time reuse pattern, which is characterized by $z_{u,n,i}$,
i.e. $z_{u,n,i}=1$ indicates user $u$ is served by server $n$ under pattern $i$ and
$z_{u,n,i}=0$ otherwise. Due to this single-server association constraint, the variables
$\{ z_{u,n,i} \}$ satisfy the constraint $\sum_{n \in \mathcal{N}} z_{u,n,i} = 1$.
Note the corresponding relaying strategies are also implied by $\{z_{u,n,i}\}$.

Denote the transmission power of server $n$ by $P_n$.
The signal-to-interference-plus-noise ratio (SINR) of the link between user $u$ and
server $n$ under pattern $i$ is denoted by $\gamma_{u,n,i}$ and can be derived as
\begin{equation}
\setlength{\abovedisplayskip}{3pt}
\setlength{\belowdisplayskip}{3pt}
\label{eq:gamma}
  \gamma_{u,n,i} =
  \frac{\mathds{1}\{ n\in \mathcal{A}_i \} P_n |g_{n,u}|^2}
  {\sigma^2 + \sum_{m \in \mathcal{A}_i, m\neq n } P_{m} |g_{m,u}|^2 },
\end{equation}
where $\sigma^2$ is the power of the thermal noise at the user and $g_{n,u}$
denotes the average gain of the channel from server $n$ to user $u$.
Let $c_{u,n,i}$ denote the spectral efficiency of this link.
To reflect the constraint that one DUE cannot transmit and receive simultaneously,
the spectral efficiency $c_{u,n,i}$ is set to be zero if server $(u+B)$ is active
under reuse pattern $i$. Therefore, the spectral efficiency can be expressed as
\begin{equation}
\setlength{\abovedisplayskip}{0pt}
\setlength{\belowdisplayskip}{3pt}
c_{u,n,i} = \mathds{1}\{ (u+B) \not\in \mathcal{A}_i \}  \cdot \log_2 (1 + \gamma_{u,n,i}).
\end{equation}

Note that the actual time resource allocated to user $u$ by server $n$ under pattern $i$
is $Ty_{u,n,i}$. The average data rate of user $u$ over all reuse patterns can be written
as follows.
\begin{equation}
\setlength{\abovedisplayskip}{4pt}
\setlength{\belowdisplayskip}{3pt}
\begin{split}
\bar R_u & = \frac{1}{T} \cdot W \cdot \sum_{i\in\mathcal I} \sum_{n\in\mathcal N} z_{u,n,i} \cdot T \cdot  y_{u,n,i} \cdot c_{u,n,i}\\
	& = W \cdot \sum_{i\in\mathcal I} \sum_{n\in\mathcal N} z_{u,n,i}  \cdot  y_{u,n,i} \cdot c_{u,n,i},
\end{split}
\end{equation}
where $W$ denotes the system bandwidth.
Besides, we assume that the DUEs in the network demand different data.
For a D2D relay, a portion of its received data are intended for others.
The effective data rate of user $u$ is thus given by
$R_u = \bar R_u - \tilde R_{(u+B)}$,
where $\tilde R_{(u+B)}$ denotes the data rate of server $(u+B)$, i.e.
user $u$ acting as a D2D relay, targeted to the served DUEs.
The total rate of server $n$ serving other DUEs over all reuse patterns is
given by
\begin{equation}
\setlength{\abovedisplayskip}{4pt}
\setlength{\belowdisplayskip}{3pt}
\tilde R_n = W \sum_{i\in \mathcal{I}} \sum_{u\in \mathcal{U}} z_{u,n,i} \cdot y_{u,n,i} \cdot c_{u,n,i}.
\end{equation}

The single-server association rule results in a combinatorial problem and finding an
optimal solution need exhaustive search.
In the next section, we will relax the single-server constraint and allow one user to be
served by multiple servers. This will enable low-complexity algorithms to find the optimal
time reuse profile.

\section{Number of Active Time Reuse Patterns} \label{Formulation}

In this section, we would like to maximize the PF metric of the network through
optimizing the time reuse profile. To this end, we formulate the optimization
problem (\ref{Eq:original_problem}), where
the optimization variables $\bm x$, $\bm y$, $\bm z$ are the overall time reuse
profile: $\{x_i, \forall i\}$, the overall user time allocation profile: $\{y_{u,n,i},
\forall u, n, i\}$, and the DUE association profile: $\{z_{u,n,i},\forall u, n, i\}$
respectively. Problem (\ref{Eq:original_problem}) is a large knapsack problem,
which is NP-hard to find its optimal solution.

\begin{figure}
\vspace{-5pt}
\setcounter{mytempeqncnt}{\value{equation}}
\begin{subequations}
\label{Eq:original_problem}
\begin{align}
    \underset{\bm{x},\bm{y},\bm{z}}{\rm maximize}\quad  & \sum_{u \in \mathcal U} \log (R_u)\label{Eq:PF_metric} \\
    {\rm subject\ to} \quad
    & \bar R_u =  W \sum_{i\in\mathcal I} \sum_{n\in\mathcal N} z_{u,n,i} \cdot y_{u,n,i} \cdot c_{u,n,i}, \forall  u\\
    % &  R_{u,n} =  W \sum_{i\in\mathcal I}  z_{u,n} \cdot y_{u,n,i} \cdot c_{u,n,i},  u \in \mathcal{U},  n \in \mathcal{N}\\
    &\tilde R_n =  W \sum_{i\in\mathcal I} \sum_{u\in\mathcal U} z_{u,n,i} \cdot y_{u,n,i} \cdot c_{u,n,i},  \forall  n \\
    & R_u = \bar R_u - \tilde R_{u+B}  , \  R_u > 0, \forall  u \label{Eq:Rate}\\
    &\sum_{n\in \mathcal N} z_{u,n,i} = 1,z_{u,n,i} \in \{ 0, 1\}, \forall u,n,i\\
    & x_i \geq \sum_{u\in \mathcal U} z_{u,n,i} \cdot y_{u,n,i}, \forall n,i \label{Eq:available_pattern}\\
    & \sum_{i\in\mathcal I} x_i = 1,  x_i \geq 0, \forall i, \ y_{u,n,i} \geq 0,  \forall u,n,i\label{Eq:x_y_constraints}
\end{align}
\end{subequations}
\hrulefill
\vspace{-15pt}
\end{figure}

By allowing fractional user association, we can drop the single-server association
indicator variable $\bm{z}$ and rely on $\bm y$ to imply the fractional user association
profile. Particularly, given $y_{u,n,i} > 0$, we can infer the $u$-th user is served
by the $n$-th server under pattern $i$.
The relaxed multi-server association problem is shown in (\ref{Eq:relaxed_problem}),
whose solution offers an upper bound of the problem in (\ref{Eq:original_problem}).
Since the active servers in each pattern $i$ will use all the available time resources
to optimize the network performance, the inequalities in (\ref{Eq:available_pattern})
indeed become the equalities in (\ref{Eq:available_pattern_relax}).
% \begin{figure}
% \vspace{-5pt}
\setcounter{mytempeqncnt}{\value{equation}}
\begin{subequations}
\setlength{\abovedisplayskip}{3pt}
\setlength{\belowdisplayskip}{3pt}
\label{Eq:relaxed_problem}
  \begin{align}
    \underset{\bm{x},\bm{y}}{\rm maximize}\quad
    & \sum_{u\in\mathcal U} \log (R_u) \label{Eq:object_1} \\
    {\rm subject\ to} \quad
    & \bar R_u =  W \sum_{i\in\mathcal I} \sum_{n\in\mathcal N}  y_{u,n,i} \cdot c_{u,n,i},  \forall u \label{Eq:RateReceive}\\
    &\tilde R_n =  W \sum_{i\in\mathcal I} \sum_{u\in\mathcal U}   y_{u,n,i} \cdot c_{u,n,i},  \forall n\label{Eq:RateSend}\\
    & R_u = \bar R_u - \tilde R_{u+B}  , \ R_u>0, \forall u \\
    & x_i = \sum_{u\in\mathcal U} y_{u,n,i}, \forall n,i \label{Eq:available_pattern_relax} \\
    &  \sum_{i\in\mathcal I} x_i = 1, x_i \geq 0, \forall i, ~y_{u,n,i} \geq 0, \forall u,n,i\label{Eq:constraint_pattern}
  \end{align}
\end{subequations}
% \hrulefill
% \vspace{-15pt}
% \end{figure}

The objective function in (\ref{Eq:relaxed_problem}) is concave and all the constraints
are linear. However, the amount of variables, i.e. $\{x_i\}$ and $\{y_{k,n,i}\}$,
increases exponentially with the number of servers due to the fact that we have $2^N-1$
different time reuse patterns.
For a large cooperative network, it becomes a formidable task to find the optimal solution
to the convex problem in (\ref{Eq:relaxed_problem}).
Fortunately, we can show that only a very limited number of reuse patterns need to be activated without sacrificing the network throughput performance.

\begin{pro}
\label{Pro:pattern_number}
There exists one optimal solution to the problem in (\ref{Eq:relaxed_problem}) where
at most $U$ out of $2^N - 1$ reuse patterns are active.
Specifically, we can find one time reuse profile $\bm x^*$ which maximizes the
objective function in (\ref{Eq:relaxed_problem}) and satisfies the following criterion:
\begin{equation}
\setlength{\abovedisplayskip}{3pt}
\setlength{\belowdisplayskip}{3pt}
    \label{Eq:pattern_number}
    \lVert \bm{x}^* \rVert_0 \leq U.
  \end{equation}
\end{pro}
\begin{proof}
Suppose we are given one particular optimal solution to the problem in
(\ref{Eq:relaxed_problem}), i.e. $\bm x$ and $\bm y$.
Each element of $\bm y$ can be re-written as $y_{u,n,i} = x_i \cdot\tau_{u,n,i}$.
The rates $\bar R_u$ and $\tilde R_n$ in (\ref{Eq:RateReceive}) and (\ref{Eq:RateSend})
can be expressed as
$\bar R_u = W \sum_{i\in\mathcal I} x_i \bar R_u^i$ and
$\tilde R_n  = W \sum_{i\in\mathcal I} x_i \tilde R_n^i$, where
$\bar R_u^i := W \sum_{n\in\mathcal N}  \tau_{u,n,i} \cdot c_{u,n,i}$
and $\tilde R_n^i := W \sum_{u\in\mathcal U}  \tau_{u,n,i} \cdot c_{u,n,i}$.
Define $\bar{\bm R} := [\bar R_1,\cdots,\bar R_U]^T$,
$\tilde{\bm R} := [\tilde R_1, \cdots, \tilde R_U]^T$,
$\bar{\bm R}^i := [\bar R_1^i,\cdots,\bar R_U^i]^T$,
$\tilde{\bm R}^i := [\tilde R_{B+1}^i,\cdots,\tilde R_N^i]^T$,
$\bm \Phi^i := \bar{\bm R}^{i}- \tilde{\bm R}^{i}$, and
$\bm R:= \bar{\bm R} - \tilde{\bm R} = [ R_1,\cdots, R_U]^T$.
We see $\bm R$ can be written as
\begin{equation}
\setlength{\abovedisplayskip}{3pt}
\setlength{\belowdisplayskip}{3pt}
\bm R = \left[{\bm \Phi}^1,\cdots,{\bm \Phi}^{2^N-1} \right]\bm x.
\end{equation}
By defining a set $\mathcal Q$ as $\mathcal Q := \{\bm \Phi^1,\bm \Phi^2,...,\bm
\Phi^{2^N-1}\}$, we see the vector $\bm R$ lies in ${\sf conv}(\mathcal Q)$, i.e. the
convex hull of $\mathcal Q$. Since the dimension of the vector $\bm R$ is $U$,
from the Carath\'{e}odory's theorem \cite{2006_Nonlinear_Programming},
we know the vector $\bm R$ must lie in the convex hull of $(p+1)$ affinely
independent vector points in $\mathcal Q$ with $p \leq U$.
Let $\mathcal Q'$ denote the set of those $(p+1)$ vector points and we know
$\bm R$ lies in ${\sf conv}(\mathcal Q')$, which is a $p$-simplex.
Furthermore, when the solution is optimal, the vector
$\bm R^*$ should reach the Pareto efficiency
\cite{2015_Zhuang,1997_Linear_Programming} given that the network utilization is
measured by the PF metric in (\ref{Eq:object_1}), which is concave with respect to
$\bm R$. Hence we can find one optimal solution in the face of the $p$-simplex,
which is also one $q$-simplex with $q<p\leq U$.
Note that the constraints $R_u>0,\forall u$ hold for the optimal solution $\bm R^*$ naturally since the objective function is meaningless for any non-positive $R_u$, otherwise the problem is infeasible.
As a result, we see $\bm R^*$ can be represented by a convex
combination of at most $U$ affinely independent points in ${\sf conv}(\mathcal Q)$.
In summary, there exits one optimal time reuse profile $\bm x^*$ satisfying
$\|\bm x^*\|_0 \le U$,
$\bm R^*=[\bm \Phi^1,...,\bm \Phi^{2^N-1}]\bm x^*$.
\end{proof}
\vspace{-5pt}
Proposition \ref{Pro:pattern_number} indicates that we only need to turn on no more than
$U$ reuse patterns to achieve the optimal network throughput performance instead of
activating all the $2^N-1$ feasible patterns.

\vspace{-5pt}
\section{Efficient Pattern Selection and Resource Allocation Algorithm}\label{Algorithms}

As indicated by Proposition \ref{Pro:pattern_number}, the following optimization
problem shares the same optimal objective value as the problem in
(\ref{Eq:relaxed_problem}).
\begin{subequations}
\setlength{\abovedisplayskip}{3pt}
\setlength{\belowdisplayskip}{3pt}
\label{Eq:problem_coordinate}
  \begin{align}
    \underset{\bm{x},\bm{y},\mathcal V}{\rm maximize}\quad
    & P(\bm x, \bm y, \mathcal V) = \sum_{u\in\mathcal U}\log (R_u)   \\
    {\rm subject\ to} \quad
    &  x_i = 0,  i\in \{i| \bm v_i\notin \mathcal V\}, {\rm (\ref{Eq:RateReceive})-(\ref{Eq:constraint_pattern})}\label{Eq:ZeroPat},\\
    & |\mathcal V| \le U, \label{Eq:NumPat}
  \end{align}
\end{subequations}
where $\mathcal V \subseteq \tilde{\mathcal V}$ denotes the candidate set of reuse
patterns containing at most $U$ patterns. The problem in (\ref{Eq:problem_coordinate})
is still very hard to solve.
We split it into two subproblems and solve them iteratively. In particular,
during the $(t+1)$-th iteration, we carry out the following updates.\\
\noindent{$\bullet$} {\it P1. Pattern Selection}:
Determine the pattern set $\mathcal{V}^{t+1}$ such that
$\mathcal{V}^{t+1}\subseteq\tilde{\mathcal{V}}$ and
$\lvert \mathcal{V}^{t+1} \rvert \leq U$;\\
\noindent{$\bullet$} {\it P2. Resource Allocation:}
Update the resource allocation and the corresponding association rules as
\begin{equation}
\setlength{\abovedisplayskip}{3pt}
\setlength{\belowdisplayskip}{3pt}
\label{Eq:Sub_2}
    \underset{\bm{x} \in \mathcal{G}(\mathcal V^{t+1}) }{\rm maximize}\ \  \underset{\bm{y} \in \mathcal F(\bm x)}{\rm maximize}\ \   P(\bm x,\bm y,
    \mathcal{V}^{t+1}), \ {\rm subject \ to} \ {\rm (\ref{Eq:ZeroPat})},
\end{equation}
where $\mathcal{G(V)} := \{ \bm{x} | \sum_{i\in \{j| \bm v_j\in \mathcal V\}} x_i = 1,
x_i \geq 0, \forall i\}$, and
$\mathcal F(\bm x) := \{ \bm{y} | \sum_{u} y_{u,n,i} = x_i, y_{u,n,i} \geq 0,
\forall u,n,i\}$.
The problem in (\ref{Eq:Sub_2}) is convex and can be solved efficiently by  general
convex solvers, e.g. CVX \cite{2015_CVX}.
However, the pattern selection subproblem has $\sum_{m=1}^{U}\binom{2^N-1}{m}$
possible solutions in theory and it is impractical to apply the exhaustive search
method. Next, we propose an approximated solution enjoying low-complexity.

From $\mathcal{V}^t$, we can define a set $\mathcal{V}_d^t$ as
\begin{equation}
\setlength{\abovedisplayskip}{5pt}
\setlength{\belowdisplayskip}{5pt}
\mathcal{V}_d^t :=
\{ \bm{v} \big| \bm v\in\tilde{\mathcal{V}}, \exists \bm{v}' \in \mathcal{V}^t, D_H(\bm{v},\bm{v}') \leq d \},
\end{equation}
where $d \leq N$ and $D_H(\bm{v},\bm{v}')$ denotes the Hamming distance between patterns
$\bm{v}$ and $\bm{v}'$.
Furthermore, we add a particular pattern set to the candidate pattern set $\mathcal{V}^t$ in the $(t+1)$-th iteration when $|\mathcal V^t|< U$, i.e.
$\mathcal{V}^{t+1} = \mathcal V^t \bigcup \mathcal{V}^*$ and
$\mathcal{V}^* \subseteq \mathcal V_d^t $ denotes the pattern set to be included
in the $(t+1)$-th iteration.
Inspired by the pattern selection approach in \cite{2016_Kuang} and
the Frank-Wolfe method \cite{2006_Nonlinear_Programming}, we put forward
an iterative pattern selection method according to the solution of the following problem:
\begin{equation}
\setlength{\abovedisplayskip}{3pt}
\setlength{\belowdisplayskip}{3pt}
\label{Eq:MM_update_D2D}
(\mathcal{V}^*, \bm x^*, \bm y^*) =
\arg \ \max_{\mathcal{V} \subseteq \mathcal{V}_d^t}\  \max_{\bm x \in \mathcal{G}(\mathcal{V})} \ \max_{\bm{y} \in \mathcal{F}(\bm{x})}
\nabla_{\bm{y}} {P}(\bm x^t, \bm{y}^t, \mathcal V^t)^T \bm{y} .
\end{equation}
Basically, in (\ref{Eq:MM_update_D2D}), we identify a set of reuse patterns to
activate by finding the direction providing the most predominant improvement in
the objective value under the specified constraints.
Denoting the $(u, n, i)$-th entry of the gradient vector
$[\nabla_{\bm{y}} {P}(\bm x^t,\bm{y}^t,\mathcal V^t)]_{u,n,i}$
by $p_{u,n,i}^t$, we can have
\begin{equation}
\setlength{\abovedisplayskip}{5pt}
\setlength{\belowdisplayskip}{5pt}
p_{u,n,i}^t = W \cdot c_{u,n,i} \cdot \left(  R_{u}^{-1} - \mathds{1}\{ n > B \} \cdot R_{n - B}^{-1} \right).
\end{equation}
Now we can establish the following result.
\begin{pro}
\label{Pro_choosepattern}
Among the solution to the problem in (\ref{Eq:MM_update_D2D}), the set $\mathcal{V}^*$
contains the following time reuse pattern:
\begin{equation}
\setlength{\abovedisplayskip}{3pt}
\setlength{\belowdisplayskip}{3pt}
\label{Eq:pro_2}
 \bm v_{i^*} = \arg \max_{\bm v_i\in \mathcal V_d^t}
 \sum_{n\in\mathcal N} p_{u_{n,i}^*,n,i}^t.
\end{equation}
where $u_{n,i}^* := \arg \max_u \ p_{u,n,i}^t$.
\end{pro}
\begin{proof}
Defining $\tau_{u,n,i}:=y_{u,n,i}/x_i$, the problem in (\ref{Eq:MM_update_D2D}) can be
re-written as
\begin{equation}
\setlength{\abovedisplayskip}{3pt}
\setlength{\belowdisplayskip}{0pt}
\label{Eq:inner_product}
 \max_{\mathcal{V} \subseteq \mathcal{V}_d^t}\
\underset{ \bm x \in \mathcal{G}(\mathcal{V}) }{\rm max} \sum_{i\in\mathcal I} x_i \sum_{n\in\cal{N}}
\underset{ \underset{\tau_{u,n,i}\geq 0}{\sum_u\tau_{u,n,i} = 1} }{\rm max}
\sum_{u\in\cal{U}} \tau_{u,n,i} \cdot p_{u,n,i}^t.
\end{equation}
The third maximization is a constrained linear programming (LP) problem and the
optimal $\tau_{u,n,i}^*$ is given by
\begin{equation}
\setlength{\abovedisplayskip}{3pt}
\setlength{\belowdisplayskip}{3pt}
\label{Eq:choose_user}
\tau_{u,n,i}^* = \mathds{1}\{u = u^*_{n,i}\}, {\rm where}\
u^*_{n,i} := \arg \max_u p_{u,n,i}^t.
\end{equation}
The problem in (\ref{Eq:inner_product}) thus becomes
\vspace{-0.2cm}
\begin{equation}
\label{Eq:changed_inner}
\max_{\mathcal{V} \subseteq \mathcal{V}_d^t}\
\underset{ \bm x \in \mathcal{G}(\mathcal{V}) }{\rm max}
\sum_{i\in\mathcal I} x_i \sum_{n\in\mathcal{N}} p_{u_{n,i}^*,n,i}^t.
\end{equation}
The problem in (\ref{Eq:changed_inner}) is again a constrained LP and
the solution is:
${x_i^*} = \mathds{1}\{\bm v_i = \bm v_{i^*} \},
\ \bm v_{i^*} = \arg \max_{\bm v_i\in \mathcal V_d^t} \sum_{n\in\mathcal N} p_{u_{n,i}^*,n,i}^t.$
\end{proof}

The solution to (\ref{Eq:MM_update_D2D}) during the $(t+1)$-th iteration simply
tells us to activate the particular time reuse pattern $\bm v_{i^*}$ to maximize
the improvement in the objective value. Clearly, the associated complexity in
finding the candidate pattern set $\mathcal{V}^{t+1}$ is determined by the size of
$\mathcal{V}_d^t$. We have the lowest complexity in the pattern selection subproblem
by setting $d = 1$. To enable faster convergence, we select $N$ time reuse patterns
in each iteration. Specifically, we modify the ON/OFF state of the $n$-th server
in each reuse pattern $\bm v_i\in\mathcal{V}^t$ as $\bm v_{i_{n}} :=
\bm v_{i}\oplus \bm e_n$ and define the set $\bar{\mathcal V}_n^t$ as
\begin{equation}
\setlength{\abovedisplayskip}{3pt}
\setlength{\belowdisplayskip}{3pt}
\label{Eq:pat_set}
\bar{\mathcal V}_n^t = \{\bm  v_{i_{n}}|\bm v_{i_{n}} := \bm v_{i} \oplus \bm e_n ,\bm v_i\in \mathcal V^{t}, \bm v_i\neq \bm e_n \}.
\end{equation}
Then we determine a new candidate reuse pattern $\bm v_{i^*_n}$ to activate according to
Proposition \ref{Pro_choosepattern} as
\begin{equation}
\setlength{\abovedisplayskip}{1pt}
\setlength{\belowdisplayskip}{1pt}
\label{Eq:prop2_distr}
\bm v_{i^*_n} = \arg \underset{ {\bm v_{i_n}\in \bar{\mathcal V}_n^t},\bm  v_{i_{n}}\notin \mathcal V^t }{\rm max} \sum_{\bar n\in\mathcal N} p_{u_{\bar n,i}^*,\bar n,i}^t.
\end{equation}
Now a temporary reuse pattern set for the $(t+1)$-th itervation, i.e.
$\hat{\mathcal{V}}^{t+1}$ can be derived as
\begin{equation}
\label{Eq:RenewSet}
\setlength{\abovedisplayskip}{3pt}
\setlength{\belowdisplayskip}{3pt}
\hat{\mathcal{V}}^{t+1} = \mathcal V^t \cup \{\bm v_{i^*_n}, n\in\mathcal{N}\}.
\end{equation}
Since the number of time reuse patterns in $\hat{\mathcal V}^{t+1}$ will become larger
than $U$ when we add $N$ time reuse patterns in each iteration, we perform the
following pattern trimming as well. In particular,
let $\bm x^{t+1}$ be the optimal solution to the problem in (\ref{Eq:Sub_2}) with
${\mathcal{V}}=\hat{\mathcal{V}}^{t+1}$. By defining a positive threshold
$\epsilon_1<1$, we can delete those patterns in $\hat{\mathcal{V}}^{t+1}$ with
negligible allocated resources, i.e.
\begin{equation}
\label{Eq:DeleteUseless}
\setlength{\abovedisplayskip}{3pt}
\setlength{\belowdisplayskip}{3pt}
\mathcal V^{t+1} = \{\bm v_i| x^{t+1}_i>\epsilon_1, \bm v_i\in\hat{\mathcal V}^{t+1}\}.
\end{equation}

Note that the solution obtained by solving (\ref{Eq:problem_coordinate}) allows
multi-server association. To meet the single-server association requirement in
(\ref{Eq:original_problem}), we simply associate the user to the server that gives
the largest data rate under the particular time reuse pattern $\bm v_i$, i.e.
\begin{equation}
\setlength{\abovedisplayskip}{3pt}
\setlength{\belowdisplayskip}{3pt}
\label{Eq:multi_2_single}
z_{u,n,i} = \mathds{1}\{n = n_{u,i}^*\}, \ n_{u,i}^* =
\arg \max_{n\in\mathcal N} y_{u,n,i} \cdot c_{u,n,i}.
\end{equation}

After finalizing the user association rule as in (\ref{Eq:multi_2_single}),
the original problem in (\ref{Eq:original_problem}) can be solved with
low complexity since $\bm z$ is given and the number of active patterns in
$\mathcal{V}^{t}$ is limited. Algorithm \ref{Alg:solve_MM} summarizes all the steps.

\begin{algorithm}[t]
\caption{Optimal time reuse with D2D relays.}
\begin{algorithmic}[1]
\State \textbf{Initialization:}
   $\epsilon_1=\epsilon_2=10^{-4}$, $t=1$, $\mathcal V^1 = \{ \bm{e}_1 + \bm{e}_{u+1}
|u \in \mathcal{U}\}$, and $P^0 = -1$. Solve (\ref{Eq:Sub_2}) with $\mathcal{V}^{1}$, obtain $\bm x^{1}$, $\bm y^{1}$, and $ P^{1}$;
\State \textbf{While} $|{P}^{t} - {P}^{t-1}| > \epsilon_2$ and $|\mathcal{V}^{t}| \leq U$
\State \quad Choose the active time reuse patterns $\{ \bm v_{i^*_n}, n\in\mathcal N\}$
as (\ref{Eq:pat_set})-(\ref{Eq:prop2_distr});
\State \quad Obtain $\hat{\mathcal{V}}^{t+1}$ with $\{ \bm v_{i^*_n}\}_{n\in\mathcal N}$ as (\ref{Eq:RenewSet});
\State \quad Solve (\ref{Eq:Sub_2}) with $\hat{\mathcal{V}}^{t+1}$, obtain $\bm x^{t+1}$, $\bm y^{t+1}$, and $ P^{t+1}$;
\State \quad Renew $\mathcal V^{t+1}$ with $\bm x^{t+1}$ as (\ref{Eq:DeleteUseless});
\State \quad $t = t + 1$;
\State \textbf{End}
\State Determine the single-server associations as (\ref{Eq:multi_2_single});
\State Solve (\ref{Eq:original_problem}) with $\mathcal{V}^t$ and the
dervied associations to get the optimal resource allocations and corresponding
relaying strategies.
\end{algorithmic}
\label{Alg:solve_MM}
\end{algorithm}

\begin{table}[t]
\centering
\caption{Brute-force CVX vs. proposed low-complexity solution.}
\scalebox{0.97}{
\begin{tabular}{|c|c|c|c|c|c|c|c|c|c|}
\hline
{\# DUEs}                & 10 & 11 & 12 &30  \\
\hline
($\theta_B$, $\theta_L$) & (8.63, 8.63) & (8.61, 8.51) & (7.12, 7.11) & (?, 7.34)  \\
\hline
($t_1$, $t_2$)           &  (735, 20) & (3004, 40) & (10608, 45) & (?, 160)
 \\
\hline
\end{tabular}}
\label{Tab:error}
\vspace{-0.5cm}
\end{table}

\vspace{-0.3cm}
\section{Numerical Results} \label{Simulation}
\vspace{-0.2cm}

In this section, we test our algorithm by simulating a cooperative network with only
one pico-BS ($30$dBm transmission power (TxPwr)) since we focus on the relay behavior
of the DUEs ($20$dBm TxPwr). See also Fig. \ref{Fig:D2D_topology} for one example.
Other parameters are set as follows.\\
\noindent{$\bullet$} DUEs are uniformly dropped in a square specified by $[0,200]$m$
\times$ $ [0,200]$m and the pico-BS is deployed at the center;\\
\noindent{$\bullet$} System bandwidth is $20$MHz. Noise PSD is $-174$dBm/Hz. The
path-loss is determined as $37.6 \log_{10} d_m + 35.3 + 5n_w $(dB)
\cite{2016_Ma_HenNet_D2D}, where $d_m$ is the distance in meters, and $n_w$ stands for
the number of walls \cite{2009_Winner}.

\begin{figure}[t]
\centering
\epsfig{file = 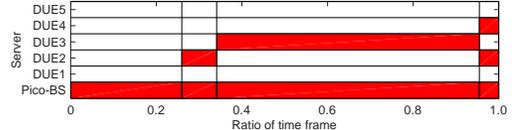, width = 0.39\textwidth}\\
\vspace{-0pt}
\caption{The active time reuse patterns obtained from the proposed algorithm with respect to the network in Fig.~\ref{Fig:D2D_topology}.  Solid color indicates that the particular time slot is occupied by the corresponding server.}
\label{Fig:Pat_UEA}
\vspace{-5pt}
\end{figure}

\begin{figure}[t]
\centering
\epsfig{file = 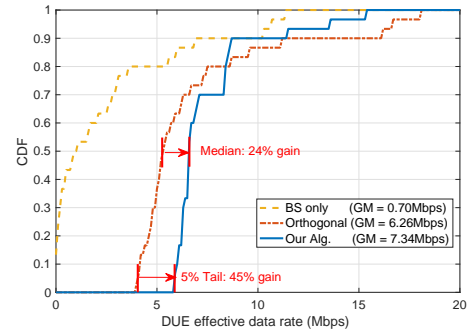, width = 0.39\textwidth}\\
\vspace{-0pt}
\caption{The CDF of $30$ DUE effective data rates. Our Alg.: solution with the proposed algorithm. Orthogonal: solution with the orthogonal scheme \cite{2016_Mastronarde_D2D_relay}. BS only: solution without D2D relays.}
\label{Fig:CDF}
\vspace{-0.5cm}
\end{figure}

Table~\ref{Tab:error} compares the geometric mean (GM) of the DUEs' throughputs and the
CPU running time between the brute-force optimal solution
(solve the problem in (\ref{Eq:relaxed_problem}) directly with CVX) and our proposed
low-complexity solution. In Table~\ref{Tab:error},
$\theta_B$ ($\theta_L$) denote the GM throughput in Mbps with the brute-force solution
(our proposed algorithm) and $t_1$ ($t_2$) is the consumed CPU time of the brute-force
solution (our proposed algorithm) in seconds. Question mark ``?'' indicates the
brute-force CVX solver can not be solved with our lab computer.
Our proposed algorithm performs satisfactorily and only suffers a very small amount
of degradation as we have $12$ DUEs. However, note that the brute-force method becomes
impractical even for a middle-sized network, e.g. a network with $30$ DUEs.

For a network with $5$ DUEs, Fig.~\ref{Fig:Pat_UEA} shows the set of activated time
reuse patterns and the fraction of total time allocated to each pattern.
The total number of active time reuse patterns is $4$, which follows
Proposition~\ref{Pro:pattern_number}.
It also indicates that the optimal time reuse patterns are not as those proposed in \cite{2013_Zhou_D2D_relay,2013_Xu_Cooperation,2016_Mastronarde_D2D_relay,2015_Li_MMM}.
The orthogonal scheme in Fig.~\ref{Fig:CDF} considers the time reuse patterns where
only the BS and one D2D relay are active \cite{2016_Mastronarde_D2D_relay},
i.e. $\mathcal V = \{ \bm{e}_1 + \bm{e}_{u+1}|u \in \mathcal{U}\}$.
In Fig. \ref{Fig:CDF}, it is also worth noting that the
orthogonal scheme could provide higher effective rates to some DUEs than our scheme.
In fact, this indicates that our proposed algorithm will ask those DUEs with high data
rates to serve as D2D relays to help the DUEs with poor channel conditions.
It is clear that our proposed algorithm achieves higher GM data rate than the other two
existing schemes.

\vspace{-0.3cm}
\section{Conclusions} \label{Conclusion}
\vspace{-0.2cm}

In this paper, we have studied the optimal time reuse patterns and the corresponding relaying
strategies in a cooperative network with D2D relays. The original optimization problem
is of a formidable size. This is due to the fact that the total number of feasible time reuse
patterns scales exponentially with the number of nodes in the network. To circumvent this dilemma,
we have shown that we just need to turn on a limited number of time reuse patterns without
sacrificing the network performance firstly. In particular, we have proved that the number of active
reuse patterns can be no more than the number of users. Secondly, we have put forward a
low-complexity algorithm to identify the small set of active time reuse patterns and solve the
large scale optimization problem approximately. Compared to those existing schemes with orthogonal
resource allocation constraints, our proposed scheme offers significant gains.

\end{document}